\documentclass[12pt]{article}
\usepackage{enumerate}
\usepackage{cite}
\usepackage{mathtools}
\usepackage{appendix}
\usepackage{amsfonts}
\usepackage{amsthm,amsmath,amssymb}
\usepackage{algpseudocode} 
\usepackage{algorithm}
\usepackage{mathrsfs}
\usepackage{fullpage}

\newtheorem{theorem}{Theorem}[section]

\newtheorem{lemma}[theorem]{Lemma}

\newtheorem{observation}[theorem]{Observation}
{\theoremstyle{remark} }
{\theoremstyle{definition} \newtheorem{definition}[theorem]{Definition}}

\newcommand{\COFL}{\textsf{COFL-Line}}
\newcommand{\DCOFL}{\textsf{COFL-Line-Dec}}
\newcommand{\SCOFL}{\textsf{COFL-Line-Sq}}
\newcommand{\DSCOFL}{\textsf{COFL-Line-Sq-Dec}}
\newcommand{\CCOFL}{\textsf{COFL-Circ}}
\newcommand{\DCCOFL}{\textsf{COFL-Circ-Dec}}
\newcommand{\MOFL}{\textsf{MOFL}}
\newcommand{\SMOFL}{\textsf{Simple-MOFL}}

\renewcommand{\circle}{\mathbb{C}}
\newcommand{\Finterval}{\mathcal{I}}

\newcommand{\Index}{\text{index}}

\date{}

\title{Efficient Algorithms for Obnoxious Facility Location on a Line Segment or Circle}
\author{Bowei Zhang}

\begin{document}
\makeatletter
\maketitle %
\begin{abstract}
    We study different restricted variations of the obnoxious facility location problem on a plane. 
The first is the constrained obnoxious facility location on a line segment (\COFL) problem. 
In this problem, we are given a line segment $\overline{pq}$, a set $P=\{p_1,p_2,\cdots,p_n\}$ of $n$ points in the plane and a given integer $k$.
Our goal is to pack $k$ maximum-radius congruent disks that are centered on $\overline{pq}$ and do not include any points in $P$. 
We provide an efficient algorithm for this problem
that executes in $O(n ^ 2 \log k + n \log k \log  (n^2 + k))$ time. 
Our result improves on the best known result of 
$O((nk)^2 \log(nk) + (n + k) \log (nk))$ time
obtained by Singireddy and Basappa\cite{singireddy2022dispersing}.
We also study the same problem where the facilities 
must be placed on a given circle
(the constrained obnoxious facility location on a circle (\CCOFL) problem).
We provide an efficient algorithm for this problem
that executes in $O(n ^ 2 \log k + n \log k \log  (n^2 + k))$ time. 
Our result improves on the best known result of 
$O((nk)^2 \log(nk) + (n + k) \log (nk))$ time
obtained by Singireddy and Basappa\cite{singireddy2022dispersing}.
The third problem we study
is the min-sum obnoxious facility location (\MOFL) problem.
In this problem, we are given a line segment $\overline{pq}$ 
and a set $P=\{p_1,p_2,\cdots,p_n\}$ of $n$ weighted points in the plane. For a given integer $k$ and real number $\lambda>0$, our goal is to pack $k$ non-overlapping congruent disks of radius $\lambda$ that are centered on $\overline{pq}$ such that the sum of the weights of the points in $P$ covered by the union of these $k$ disks is minimized. We provide an efficient algorithm
that executes in $O(nk\cdot \alpha(nk) \log^3 {nk})$ time, where $\alpha(.)$ is the inverse Ackermann function.
The best known previous result is an $O(n^3k)$ time obtained by 
Singireddy and Basappa\cite{singireddy2022dispersing}.
\end{abstract}

\section{Introduction}

The study of facility location problems is an important branch of operations research and algorithm design. Such problems are typically concerned with finding optimum locations for facilities to serve a set of demand points (clients). In general, we would like to place facilities such that they are close to demand points \cite{hochbaum1985best}. The most common objectives include minimizing the total service costs(i.e., the uncapacitated facility location problem(UFLP)) \cite{hajiaghayi2003facility} or covering all demand points within 
a certain radius (i.e., the $k$-center problem) \cite{hochbaum1985best}. The mathematical model of the facility location problem can be applied to several other problems, including clustering and the lock-box problem. 
However, in certain applications, facilities can be obnoxious and must be placed as far as possible away from other facilities such as hospitals, fire stations, post offices, schools,  \cite{church1978locating}, and other obnoxious facilities. 
In such applications, such facilities should typically be located on the sides of highways because of heavy transportation requirements. 
This motivates the problems we study in this paper, that is how to place
obnoxious facilities on a line segment or circle.

First, we review existing models for 
obnoxious facility placement. 
Church and Garfinkel \cite{church1978locating} first introduced the obnoxious $p$-median problem. 
In this problem, the positions of the non-obnoxious facilities are given.
The goal is to locate $p$ obnoxious facilities to maximize the cumulative minimum distance from the non-obnoxious facilities to these $p$ obnoxious facilities. 
The obnoxious $p$-median problem has been proved to be NP-hard \cite{tamir1991obnoxious}. 
Herr{\'a}n, Alberto and Colmenar\cite{herran2020parallel} demonstrated that
the obnoxious $p$-median problem can be formulated as an integer linear program and obtained the best solution in 137 of the 144 instances in the benchmark.
Drezner and Wesolowsky \cite{drezner1995obnoxious} provided
another formulation of the obnoxious facility location problem.
Their goal was to locate an obnoxious facility that is as far as possible from the arcs and nodes of a given network. 
They provided an $(1 - \epsilon)$-approximation algorithm that executes in $O(m^3 \log (1/ \epsilon ))$ time for the weighted version of the problem, where $m$ is the number of arcs in the network. 
Michael \cite{segal2003placing} improved the execution time to $O(m^2 \log n \log (1/ \epsilon ))$ by modifying the network to a rectilinear network, where $n$ is the number of nodes in the network. 
Singireddy and Basappa\cite{singireddy2022dispersing} provided a
$O((nk)^2 \log(nk) + (n + k) \log (nk))$ time algorithm for the constrained obnoxious facility location problem, where $n$ is the number of non-obnoxious and $k$ is the number of the obnoxious facilities.

Another popular variation is the {\em minimum-sum obnoxious facility location} (\MOFL) problem. We are given a set of weighted points corresponding to
non-obnoxious facilities. We must place $k$ obnoxious facilities such that they minimize the total weight of the points covered. 
The \MOFL\ problem is motivated by applications where we must place a small number of obnoxious facilities and minimize the total weight of the non-obnoxious facilities that can be influenced. 
These obnoxious facilities influence the area around them, where each area can be approximated by a disk (or other shapes). 
Dreznre and Wesolowsky \cite{drezner1995obnoxious} first studied this problem where the requirement was to place a single facility. 
They modeled the area influenced by this obnoxious facility as a rectangle or disk and provided an algorithm for solving both cases in $O(n^2)$ time. Singireddy and Basappa \cite{singireddyconstrained} improved this to $O(n\log n)$ time and presented a dynamic
programming solution that solves the case of placing  $k$ obnoxious facilities using $O(n^3k)$ time.

\subsection{Problem Definitions}

In this study, we focus on the constrained obnoxious facility location problem,
on either a line segment (\COFL) or circle (\CCOFL),in addition to the  \MOFL\ problem.
For a point $p\in \mathbb{R}^2$, we use $B(p, r)$ to denote a disk
of radius $r$ centered at $p$.
We use $d(p,q)$ to denote the distance between point $p$ and point $q$.
We use $[n]$ to denote the set $\{1,\ldots, n\}$.

Next, we formally define these problems.

\begin{definition}
Constrained Obnoxious Facility Location on a Line (\COFL) Problem.
We are given a set $P = \{p_1,p_2,\cdots, p_n\}$ of $n$ demand points in the plane, a line segment $\overline{pq}$, and a positive integer $k$.
Without loss of generality, we can suppose that $\overline{pq}$ lies on the x-axis.
Our goal is to place $k$ facilities (also called centers) $C=\{c_1,\ldots, c_k\}$ on $\overline{pq}$
such that these (obnoxious) facilities are far away from the demand points 
and from each other.
Formally, our goal is to maximize the radius $\lambda$ defined as
$$
\lambda=\min\{\min_{i\in [n]}\min_{j\in [k]} d(p_i, c_j), 
\alpha \cdot \min_{i\in [k-1]} d(c_i, c_{i+1})
\}
$$
where $\alpha>0$ is a given fixed coefficient (see Figure \ref{fig:cofl-line} for an example).
\end{definition}

\begin{center}
 \begin{figure}[t]
  \includegraphics[width=1\textwidth]{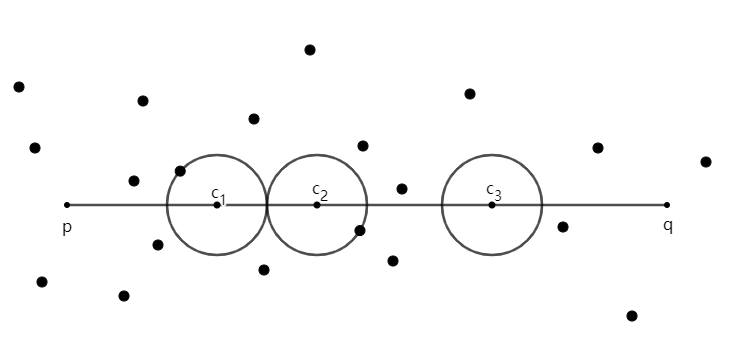}
  \caption{Optimal solution for example of \COFL\ when $k=3$ and $\alpha = 1/2$. $c_1,c_2,c_3$ are the centers that we place. The larger black nodes are the demand points. }
  \label{fig:cofl-line}
\end{figure}
\end{center}

Note that our definition is marginally more general than that defined by Singireddy and Basappa\cite{singireddy2022dispersing}.
In particular, when $\alpha=1/2$, our problem is to place 
$k$ non-overlapping disks of the same radius $r$ centered on the line segment $\overline{pq}$, where no point in $P$ is contained in any disk,
and the the radius $r$ is maximized.

We also consider the following variant
where a disk is replaced by a square
(i.e., we measure distance using $\ell_\infty$ norm instead
of $\ell_2$ norm).

\begin{definition}
(Constrained Obnoxious Facility Location with Squares  (\SCOFL))
The input of \SCOFL\ is exactly the same as that of \COFL. 
The only difference is that we want to pack $k$ maximum-size axis-aligned non-overlapping squares centered on $\overline{pq}$ such that no point of $P$ lies inside any of these squares. The size of a square is defined as its side length. 
\end{definition}

In this study, we first solve \SCOFL, which is simpler than \COFL.

We next define the \CCOFL. The problem is virtually the same as \COFL, except that
the facilities must be placed on a predetermined circle, instead of a line 
segment.

\begin{definition}
Constrained Obnoxious Facility Location on a Circle (\CCOFL).
We are given a set $P=\{p_1, p_2, \cdots, p_n\}$ of $n$ demand points in the plane, a predetermined circle $\circle$ with radius $r_c$, and a positive integer $k$.
We must locate $k$ facility sites on circle $\circle$. 
Our goal is to maximize the radius $\lambda$ defined as
$$
\lambda=\min\{\min_{i\in [n]}\min_{j\in [k]} d(p_i, c_j), 
\alpha \cdot \min_{i\in [k]} d(c_i, c_{i+1})
\}
$$
where $\alpha>0$ is a given fixed coefficient
and $c_{k+1}$ is understood as $c_1$(see Figure \ref{fig:ccofl} for an example).
\end{definition}

   \begin{center}
 \begin{figure}[h]
\centerline{  \includegraphics[width=0.8\textwidth]{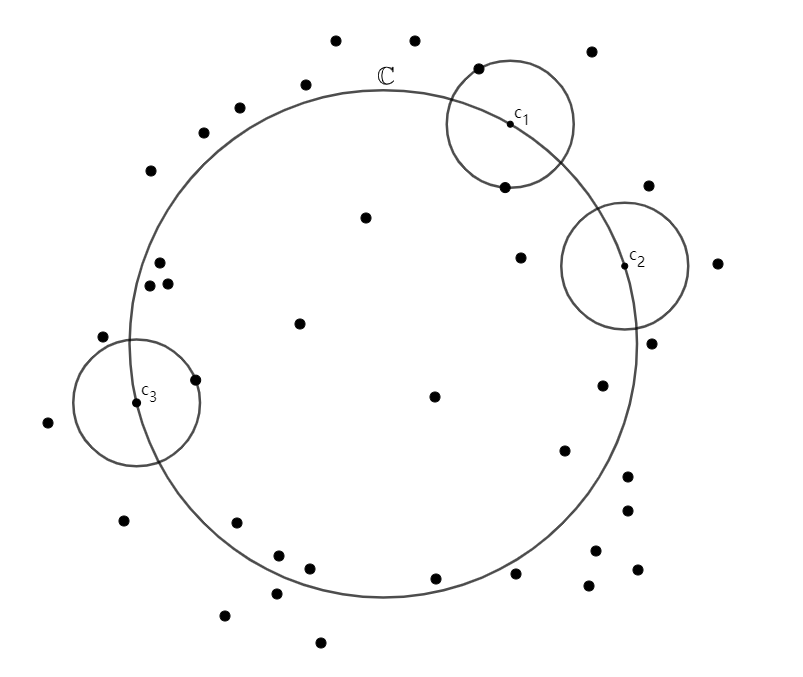}}
  \caption{Optimal solution for example of \CCOFL\ when $k=3$ and $\alpha = 1$. $c_1,c_2,c_3$ are the centers that we place. The larger black nodes are the demand points.}
  \label{fig:ccofl}
\end{figure}
\end{center}


Next, we formally define the \MOFL\ problem. In this problem, the impact radius of each facility is fixed,
and our objective is to minimize the total impact.

\begin{definition}
Minsum Obnoxious Facility Location (\MOFL) Problem.
We are given a set $P = \{p_1,p_2,\cdots, p_n \}$  of $n$ demand points in the plane with weight $\{w_1,w_2,\cdots,w_n\}$, a line segment $\overline{pq}$ and positive integer $k$, and a positive real number $\lambda > 0$. 
Our goal is to pack $k$  disks $D_1,\cdots,D_k$ of radius $\lambda$ centered on $\overline{pq}$ such that 
$$
\sum_{j=1}^k \sum_{i : p_i\in D_j} w_i
$$
is minimized 
(i.e., the total weight of the points in $P$ covered by the $k$  disks is minimized)
under the condition that $\min\limits_{i\in [k]}\{d(c_i, c_{i+1})\} \ge \alpha\cdot\lambda$ ($c_{k+1}$ is understood as $c_1$ here), where $\alpha>0$ is a given positive constant(see Figure \ref{fig:mofl} for an example).
\end{definition}

\begin{center}
 \begin{figure}[h]
  \includegraphics[width=1\textwidth]{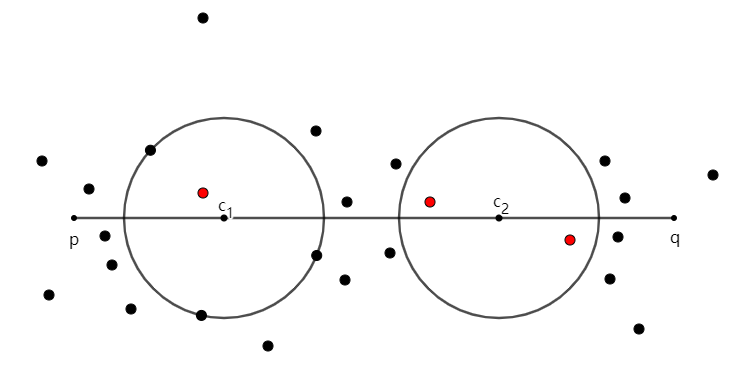}
  \caption{Optimal solution for example of \MOFL. The weight of all the demand nodes is set as ``$1$”. $c_1,c_2$ are centers that we place. The three red nodes are the demand points that are covered by the disks. The minimum total weight of the demand points that are covered by disks is three.  }
  \label{fig:mofl}
\end{figure}
\end{center}

\subsection{Previous Results and Our Contributions}

In this section, we summarize the previous results and our contributions to the above problems.

\begin{enumerate}
    \item \SCOFL\ : 
    We first consider the \SCOFL\ problem, which is easier than the \COFL\ problem. Singireddy and Basappa \cite{singireddy2022dispersing}  provided an $O((nk)^2 \log(nk) + (n + k) \log (nk))$ exact algorithm for the \SCOFL\ problem. 
    We present an $O(n \log n \log (n + k) )$ time algorithm for this problem, using the matrix search technique developed in \cite{2011Representing, chen2013algorithms}.
    
    \item \COFL\ : 
    For the $\COFL$ problem, Singireddy and Basappa\cite{singireddyconstrained} first designed a $(1-\epsilon)$-factor approximation algorithm that executes in $O((n+k) \log (\frac{||pq||}{2(k-1)\epsilon}))$ time, where $k$ is the number of obnoxious facilities and $pq$ is the segment on which the obnoxious facilities can be placed. Then, Singireddy and Basappa \cite{singireddy2022dispersing} improved their algorithm and provided an $O((nk)^2 \log(nk) + (n + k) \log (nk))$ exact algorithm based on a binary search on all candidates. 
    They also claimed another $O((n + k)^2)$ time algorithm using parametric search.
    However, their algorithm is incorrect because they do not consider the interference of adjacent intervals (i.e., placing a center in one interval could influence the placement in the adjacent interval, and hence should not be considered independently). 
    
   For the \COFL\ problem, we first consider the corresponding following decision problem:
   for a fixed radius $\lambda$ of the disks, compute the maximum number of centers that we can place such that no demand point is contained in any of these disks.
   We demonstrate that the decision problem can be solved in $O(n \log n)$ time. 
   Then, using the matrix search technique
   developed in \cite{2011Representing, chen2013algorithms}, we provide an $O(n^2 \log k + n \log k \log  (n^2 + k))$ time algorithm for \COFL.
    
    \item \CCOFL\ : 
    The \CCOFL\ problem is similar to the \COFL\ problem. The main difference is 
    that the decision algorithm for \CCOFL\ problem is more difficult to solve. 
    A naive solution that enumerates all starting points in the circle would require $O(n^2\log n)$ time. Singireddy and Basappa \cite{singireddy2022dispersing} offered a $(1-\epsilon)$-approximation algorithm for the \CCOFL\ problem in $O(n(n+k) \log (\frac{||pq||}{2(k-1)\epsilon}))$ time by executing their $(1-\epsilon)$-approximation algorithm for \COFL\ problem $n$ times (each for a different starting point).
    We design an efficient $O(n \log n \log k)$ time algorithm for the decision version of \CCOFL\ problem using a persistent segment tree.
    For the optimization problem  \CCOFL, we present an $O(n ^ 2 \log k + n (\log^2 n + \log k) \log(n^2 + k) )$ time exact algorithm.

    \item \MOFL\ : 
    For the \MOFL\ problem,  Singireddy and Basappa~\cite{singireddyconstrained} first
    presented a dynamic programming solution that executes in $O(n^3k)$ time, where $k$ is the number of obnoxious facilities. For the special case $k=1$, they provided an $O(n\log n)$ time algorithm.
    
    For the \MOFL\ problem, we demonstrate that we can transform it to a minimum-weight $k$-link path problem \cite{aggarwal1993finding}. We prove that the number of possible positions of the centers
    can be limited to $O(nk)$, and provide an 
    $O(nk\cdot \alpha(nk) \log^3 {nk})$ time algorithm for this problem.
\end{enumerate}

We summarize these results in the Table \ref{result}.
\begin{table}[]
\caption{Previous results and our results}
\label{result}
\begin{tabular}{|l|l|l|}
\hline
             & Singireddy et al.’s result                                                                                                        & Our result                                    \\ \hline
COFL-Line-Sq & $O((nk)^2 \log(nk) + (n + k) \log (nk))$                                                                                          & $O(n \log n \log (n + k) )$                   \\ \hline
COFL-Line    & $O((nk)^2 \log(nk) + (n + k) \log (nk))$                                                                                          & $O(n^2 \log k + n \log k \log  (n^2 + k))$    \\ \hline
COFL-Circ    & \begin{tabular}[c]{@{}l@{}}$O(n(n+k) \log (\frac{||pq||}{2(k-1)\epsilon}))$\\ $(1-\epsilon)$-approximation algorithm\end{tabular} & $O(n ^ 2 \log k + n (\log^2 n + \log k) \log(n^2 + k) )$ \\ \hline
MOFL         & $O(n^3k)$                                                                                                                         & $O(nk\cdot \alpha(nk) \log^3 {nk})$           \\ \hline
\end{tabular}
\end{table}

\subsection{Other Related Work}
Many variations of facility location problem including obnoxious facility location problem have been studied in the literature. 

\subsubsection{Related Obnoxious Facility
Location Problem}
Most papers about obnoxious facility location problem are modeled as chosen locations for the obnoxious facilities from a given set. 
But consider the locations of obnoxious facilities in reality, it can be far from the clients and from other obnoxious facilities and may not be limited to a set of potential locations.
We introduce literature such that consider models locating facilities in a given region and using euclidean distance here.
Shamos and Hoey\cite{1975Closest} introduce the first result of single obnoxious facility models. They find the largest circle that does not cover any points in a given set. However, they only take it as a geometrical problem. The center of the circle is the optimal location of an obnoxious facility and the radius of the circle is the maximal shortest distance from the obnoxious facility to any points in the given set.
Melachrinoudis and Cullinane\cite{mela1985} considered adding an extra obnoxious facility when some obnoxious facilities already exist under the condition that the extra obnoxious facility must be located outside circles centered at the existing facilities. 
D{\'\i}az-B{\'a}{\~n}ez et al.\cite{diaz2006locating} analyzed the problem in three-dimensional space such that
placing an “obnoxious” plane maximizes the minimum distance to a given set of communities. Suzuki et al.\cite{https://doi.org/10.48550/arxiv.2008.04386}
extend the single obnoxious facility model with weighted distance and give an optimal solution algorithm.

Drezner et al.\cite{drezner2019planar} provided a solution approximate to the multiple obnoxious-facilities problem based on Voronoi points. The problem is to place $p$ obnoxious facilities such that maximize the shortest distance between communities and facilities. The distance between each pair of facilities must be at least $D$. 
A similar formulation was proposed by Welch et al.\cite{welch2006multifacility}, who gave a solution that proposed optimality under evaluation on a set of randomly generated problems of up to five facilities and 120 communities.

\subsubsection{Related Facility Location Problem}
The metric uncapacitated facility location problem(UFLP) is the most basic facility location problem. UFLP has many applications in a large number of settings\cite{cornuejols1983uncapicitated} and also can be applied to more complicated location models. In the UFLP, We are given a set $F$ of the location of facilities, a set $C$ of clients. The cost for opening facility at location $i \in F$ is $f_i$. The cost for connecting client $j \in C$ to a facility that locate at $i \in F$ is $c_{ij}$. 
Our goal is to choose a subset of the locations for opening facilities in $F$, and connect each client in $C$ to an open facility so that the total cost for opening facilities and connecting each client to those facilities is minimized.
After Shmoys, Tardos and Aardal \cite{shmoys1997approximation} introduce the first constant factor approximation algorithm for UFLP, many constant factor approximation algorithms have been proposed. We summarize those results in table \ref{tab-1}.
\begin{table}[!h]
\caption{Result of approximation algorithm for UFLP(approx. factor is short for approximation factor)}
\label{tab-1}
\begin{tabular}{|l|l|l|l|}
\hline
approx. factor & technique                                   & running time               & reference         \\ \hline
$O(\ln n_c)$         & greedy algorithm                            & $O(n^3)$                   & Hochbaum\cite{hochbaum1982heuristics}          \\
$5+\epsilon$         & local search                                & $O(n^6 \log(n/ \epsilon))$ & Korupolu et al.\cite{korupolu2000analysis}   \\
3                    & primal-dual method                          & $O(n^2\log n)$             & Jain and Vazirani\cite{jain2001approximation} \\
1.861                & greedy algorithm                            & $O(n^2\log n)$             & Mahdian et al.\cite{mahdian2001greedy,jain2003greedy}    \\
1.853                & primal-dual method with greedy & $O(n^3)$                   & Charikar and Guha\cite{charikar2005improved} \\
1.61                 & greedy algorithm                            & $O(n^3)$                   & Jain et al.\cite{jain2002new,jain2003greedy}       \\
1.52                 & greedy algorithm with cost scaling          & $\tilde{O}(n)$             & Mahdian et al.\cite{mahdian2006approximation}\\
\hline 
\end{tabular}
\end{table}

\section{Preliminaries}
\subsection{Parametric Search}
\label{psearch}
In the design and analysis of algorithms for combinatorial optimization,  Megiddo\cite{megiddo1983applying} introduced parametric search as a technique that transforms a decision algorithm (if the optimization problem holds constraints for a given value) into an optimization algorithm (find the optimal solution). It is commonly used to solve optimization problems in computational geometry.

The basic idea of a parametric search is to simulate a test algorithm that takes numerical parameter X as the input. We suppose this test algorithm takes the optimal solution $X^{*}$ as its input. 

This test algorithm should be discontinuous when $X=X^{*}$. We only check the parameter $X$ by a simple comparisons of $X$ with other given values or test the sign of low-degree polynomial functions of X(could be generated by an observation).
Then, we must simulate each of these comparisons or tests when the value of $X$ is unknown. 
Therefore, we require another decision algorithm (denoted as the second algorithm) to simulate each comparison. The second algorithm uses another numerical parameter $Y$ as input and determines if $Y > X^*$, $Y < X^*$, or $Y = X^*$.

The second algorithm itself also can be used as the test algorithm for it is actually discontinuous at $X^{*}$, whereas in other applications we use other test algorithms (a comparison sorting algorithm is commonly used). In the advanced versions of the parametric search technique, we use a parallel algorithm as the test algorithm. Because we can group the comparisons that must be simulated into batches, we can significantly reduce the number of instantiations of the decision algorithm.

Megiddo\cite{megiddo1983applying} introduced a parallel sorting scheme that can be used for parametric searches.
We denote the execution time of the second algorithm(decision algorithm) $A$ as $T_A$. The parallel version of $A$, denoted by $A_p$, uses $P$ processors and executes in $T_p$ parallel steps. Then, we can use the binary search approach to resolve the comparisons in each parallel step. The total cost of the parametric search using this parallel sorting scheme is $O(PT_P+T_PT_A\log P)$ time.

\subsection{Matrix Search}
Wang et al. \cite{2011Representing} introduced the technique of binary search on sorted arrays, which we call matrix search in our work. 
The matrix search algorithm essentially similar to the linear-time selection algorithm \cite{leiserson1994introduction}.
Here, we provide an overview of the concept of matrix search.
First, we provide the necessary definitions for matrix search.
\begin{definition}Feasible Value, Feasibility Test \cite{jain2021algorithms}

Suppose we have a monotone decreasing function $f(x)$ and $f(\lambda^*) = 0$.
Given any $\lambda$, the decision problem
is to determine if $\lambda \ge \lambda^*$, i.e., if $f(\lambda) \ge 0$. This type of decision problem is called a feasibility test. If $\lambda \ge \lambda^*$, we say that $\lambda$ is a feasible value.
\end{definition}

\begin{definition}Matrix Search \cite{jain2021algorithms}

Given a set of $M$ sorted arrays $\{A_1,A_2,\cdots,A_M\}$, such that each array's size is at most $N$ and each array element can be evaluated in O(1) time, we must determine the smallest feasible value in these arrays.
\end{definition}

Next, we introduce the concept of matrix search. We choose a constant number of elements as “representative elements” from each array $A_i$($i \in [M]$). 
Then, we compute the (weighted) median of these $O(M)$ representative elements, denoted by $m_A$. We call the feasibility test to determine if $m_A \ge \lambda^*$  ($m_A$ is a feasible value), after which half of the representative elements can be removed. 

Then, we can carefully choose the representative elements such that a constant fraction of the elements in all $M$ arrays can be removed. 
We apply the above procedure recursively to the remaining elements. 
After $O(\log(N + M))$ iterations, the smallest feasible value can be found. 
In each iteration, we must compute the $O(M)$ representative elements and their (weighted) median and execute the feasibility test once, which requires $O(\log(N + M))$ feasibility tests and $O(M \log N)$ time, excluding the feasibility tests. 

\subsection{Persistent Segment Tree}
First, we briefly review segment tree (see \cite{bentley1980optimal}). 
 A segment tree is a data structure that stores information in an array as a tree. 
 This structure allows efficient answering of range queries over an array and yet continues to allows quick modification of the array. It requires $O(N)$ time and space complexity for building a segment tree, where $N$ is the length of the array. It supports finding the minimum or sum of any range of consecutive array elements (called a query) in $O(\log N)$ time. It also allows us to modify the array online by adding a value to an array element or modifying the values of a range (e.g., assigning a value to all elements, or adding a value to all elements in a range). 
 
 We introduce an example of the problem that a segment tree can solve. 
 We are given an array of $N$ values $a[0],a[1],\cdots,a[N-1]$. 
  Without loss of generality, we can assume that $N=2^n$. The following two operations should be supported by a segment tree in $O(\log N)$ time:
  \begin{enumerate}
      \item 
      \textbf{SUM} : for given $i, j$, compute $\sum\limits_{t=i}^j a[t]$.
      \item 
      \textbf{UPDATE} : for given $x, v$, update $a[x] \leftarrow a[x] + v$.
  \end{enumerate}
   We can build a segment tree using recursion. Every time we store the sum of the current range of the array in the corresponding node (also called the value of the node), we divide the current range of the array into two halves (if the length of the range is greater than one). We perform this recursively on both halves until the length of the current range is one (see Figure \ref{fig:tree-1} for an example). The node corresponding to the entire range of array $A$ is called the root of the segment tree.
   
   For the \textbf{SUM} operation, we can obtain the result by traversing the root of the segment tree. There are three situations for the node that we search.

  \begin{enumerate}
  \item
  If the range of the current node does not intersect the given range, then do nothing.
  \item
  If the range of the current node partially overlaps the given range, then traverse its children.
  \item
  If the range of the current node is within the given range, its value is added to the result.
  \end{enumerate}
  
  The \textbf{UPDATE} operation can also be performed by traversing the roots of the segment tree. We add $v$ to the value of all nodes such that given index $x$ is in their range (see Figure \ref{fig:tree-2} for an example).
  We provide the pseudo code for building, \textbf{SUM} and \textbf{UPDATE} of the segment tree in the following.
  
  \begin{algorithm}
	\caption{build$(root, l = 0, r = N-1)$} 
	\begin{algorithmic}[1]
	    \State create current node with range $[l, r]$
	    \If{$r - l$ equals $1$} 
	    \State \textbf{return}
	    \EndIf
	    \State $mid = \lfloor (l+r)/2 \rfloor$
	    \State build(left child of the current node, l, mid)
	    \State build(right child of the current node, mid + 1, r)
	\end{algorithmic} 
\end{algorithm}

\begin{algorithm}
	\caption{sum$(root, i, j)$} 
	\begin{algorithmic}[1]
	\If{the range of the current node does not intersect the given range $[i, j]$} 
	\State \textbf{return} 0
	\ElsIf{the range of the current node partially overlaps $[i, j]$}
	\State \textbf{return} sum(left child of the current node, $i$, $j$) + sum(right child of the current node, $i$, $j$)
	\ElsIf{the range of the current node is within $[i, j]$}
	\State \textbf{return|} the value of the current node
	\EndIf
	\end{algorithmic} 
\end{algorithm}

\begin{algorithm}
	\caption{update$(root, x, v)$} 
	\begin{algorithmic}[1]
	    \If{the range of the current node is not in the given range} 
	    \State \textbf{return}
	    \Else 
	    \State add $v$ to the value of the current node 
	    \State update(left child of the current node, $x$, $v$)  \State update(right child of the current node, $x$, $v$)
	    \EndIf
	\end{algorithmic} 
\end{algorithm}
  
 \begin{figure}[h]
  \includegraphics[width=1\textwidth]{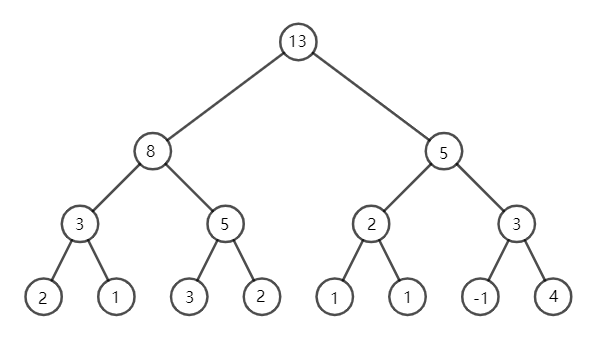}
  \caption{Instance of segment tree when $N=8$ ($A[]=\{2,1,3,2,1,1,-1,4\}$). The number in each node is the sum of its corresponding range in array $A[]$. The top node is called the root. }
  \label{fig:tree-1}
\end{figure}
 
 \begin{figure}[h]
  \includegraphics[width=1\textwidth]{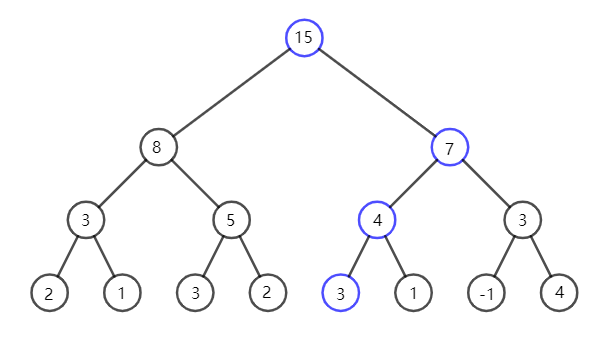}
  \caption{Example of \textbf{UPDATE} operation ($A[4] \leftarrow A[4] + 2$). The number in each node is the sum of its corresponding range in array $A[]$. The blue nodes are the nodes that are actually modified in this \textbf{UPDATE} operation.}
  \label{fig:tree-2}
\end{figure}
 
 Next, we introduce a persistent segment tree. 
 A persistent data structure is a data structure that preserves its previous version when it is modified. A data structure is partially persistent if all versions can be accessed, yet only the newest version can be modified. Persistent data structures can be used in version control applications such as Git, which enable multiple users to create new branches from the current version, make changes without modifying the older versions, and backtrack to an older version. For a detailed introduction of a persistent data structure, see \cite{kaplan2018persistent}. 
 
 A persistent segment tree is used to implement persistency in a segment tree. A persistent segment tree can preserve its past states while supporting updates. 
 
 We now provide an overview of the basic idea of a persistent segment tree.
 To preserve the previous state after each update operation, a new version of the segment tree can be built after each update operation. Suppose we have $Q$ updates in total, we would have $Q + 1$ versions of the segment tree.
Simply, we could store all the previous versions of the segment tree.
However, the building of a segment tree requires $O(N)$ time and space complexity. Thus, for $Q$ update operations it would require $O(QN)$ time and space complexity to preserve the previous state. Next, we introduce a more efficient approach to accomplish this task such that each update operation can be completed in $O(\log N)$ time and space complexity. 
The basic idea is that we only create those nodes that are actually modified in this operation in the new version of the segment tree, and share the remainder of the unchanged nodes from the previous version. This is acceptable because for each update operation, the number of nodes that are actually modified is $O(\log N)$ (see Figure \ref{fig:st}).
In conclusion, we can apply persistent segment tree to ensure that it requires $O(\log n)$ time for each query or modification.
\begin{center}
 \begin{figure}[t]
  \includegraphics[width=1\textwidth]{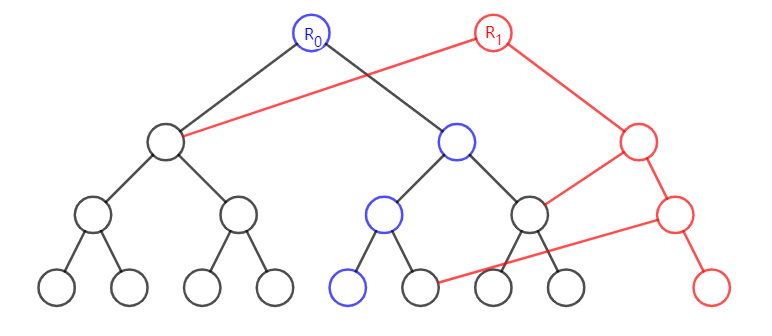}
  \caption{Example of persistent segment tree. $R_0$ is the root of the segment tree before the update operation and $R_1$ is the new root of the segment tree after the update operation.  The red nodes are the nodes that are created for an update operation. The blue nodes are the nodes that are modified in the update operation. }
  \label{fig:st}
\end{figure}
\end{center} 

\subsection{$k$-link Shortest Path}
Aggarwal \cite{aggarwal1993finding} introduced an efficient algorithm for the $k$-link shortest path problem with a convex or concave Monge property and included several applications such as data optimization and data compression. We only introduce the $k$-link shortest path problem with the convex Monge property.
First, we provide the definition of the $k$-link shortest path problem and convex Monge property.
\begin{definition}$k$-link shortest path problem \cite{aggarwal1993finding}.

Given $G = (V, E)$ as an edge weighted, complete, directed acyclic graph (DAG) with the vertex set $V = \{1, 2, \cdots ,n\}$. The weight of edge $(i,j)$ is $w(i,j)$($1 \le i <j \le n$).
For a path that contains exactly $k$ links (i.e., edges), we call the path a $k$-link path.
For any two vertices $i,j$, the $k$-link shortest path from $i$ to $j$ is the path from $i$ to $j$ that contains exactly $k$ links and has the minimum total weight between all such paths. Our goal is to determine the $k$-link shortest path from $1$ to $n$.
\end{definition}

\begin{definition}Convex Monge property \cite{aggarwal1993finding}.

For a weighted DAG $G$, if $w(i,j) + w(i + 1, j + 1) \ge w(i, j + 1) + w(i + 1, j)$ holds for all $1<i+1<j< n$, the DAG $G$ satisfies the convex Monge property.
\end{definition}

If the weights of the edges in a DAG $G$ satisfy the convex Monge property, the minimum $k$-link shortest path problem can be solved in $O(n \alpha(n) \log ^ 3 n)$ time \cite{aggarwal1993finding}.

Next, we provide an overview of Aggarwal's algorithm \cite{aggarwal1993finding}.
Let $G=(V,E)$ be the given weighted, complete DAG with the vertex set $V=\{1,2,\cdots,n\}$. We define DAG $G(\tau)$ as the DAG with the same sets of edges and vertices as G, while adding $\tau$ to all edge weights in $E$ (the weight for edge $(i,j)$ is $w(i,j)+\tau$ in $G(\tau)$). We can prove that if the minimum weight path from $1$ to $n$ in $G(\tau)$ has $k$ links, then this path is
the minimum weight $k$-1ink path from $1$ to $n$ in $G$ and the number of links in the minimum weight path from $1$ to $n$ is nonincreasing as $\tau$ increases \cite{aggarwal1993finding}. Therefore, we can solve this problem using a binary search on $\tau$. We improve it by using the parametric search paradigm given by Megiddo \cite{megiddo1983applying} (introduced in Section~\ref{psearch}).
Grossberg \cite{grossberg1982theory} provided a parallel algorithm that can compute the minimum weight path in $G(\tau)$ in
$O(\log^2 n)$ time using $O(n)$ processors. A decision algorithm that can compute the minimum weight path in $G(\tau)$ using $O(n\alpha(n))$ time can be found in Klawe and Kleitman's work\cite{doi:10.1137/0403009}.
Now, using the parametric search paradigm \cite{megiddo1983applying} that uses Grossberg's \cite{grossberg1982theory} algorithm as the test algorithm and Klawe and Kleitman's algorithm as the second algorithm, we can solve the $k$-link shortest path problem in $O(n \alpha(n) \log ^ 3 n)$ time.

\section{Constrained Obnoxious Facility Location Problems}
\subsection{\SCOFL}


We first consider the constrained obnoxious facility location problem
on a line segment with squares (\SCOFL).
Recall that we must pack $k$ nonoverlapping axis-aligned squares of the same size centered on the given segment $\overline{pq}$ such that no demand point of $P$ lies inside any of these squares. The size of a square is defined as half of its side length. We would like to maximize the size of the squares; we denote the maximum size as $\lambda^*$. 
Without loss of generality, we can assume that $\overline{pq}$ lies on the $x$-axis.

\subsubsection{Decision Version: \DSCOFL}
First, we solve the decision version of the \SCOFL\ problem (\DSCOFL), defined as follows. Given $\lambda$, \DSCOFL\ asks if we can place $k$ nonoverlapping squares of size $\lambda$ such that no demand point of $P$ lies inside any of these squares.
In fact, we present an algorithm $A$ to compute the maximum number of squares of size $\lambda$ that can be placed. We denote this number by $A(\lambda)$. Clearly, if we can compute $A(\lambda)$ efficiently, we can solve the decision problem \DSCOFL\ using the same time.

We define the $L_\infty$ distance from a point $p$ in $P$ to $\overline{pq}$ as $d(p)$ and the coordinate of each point $p_i$ in $P$ as $(x_i, y_i)$. Clearly, a point $p\in P$ with $d(p) \ge \lambda / 2$ can be removed from $P$ because no square of size $\lambda$ centered on $\overline{pq}$ can contain $p$. After removing those points, for each remaining point $p_i\in P$, we construct an interval $[x_i-\lambda / 2, x_i+\lambda / 2]$ (called a {\em forbidden interval}). 
Note that we can not place a square centered in a forbidden interval (otherwise this square would contain $p_i$). Removing all forbidden intervals from $\overline{pq}$,
we obtain a set $\Finterval$ of feasible intervals (see Figure~\ref{fig:sq-1}).

\begin{center}
 \begin{figure}[t]
  \includegraphics[width=1\textwidth]{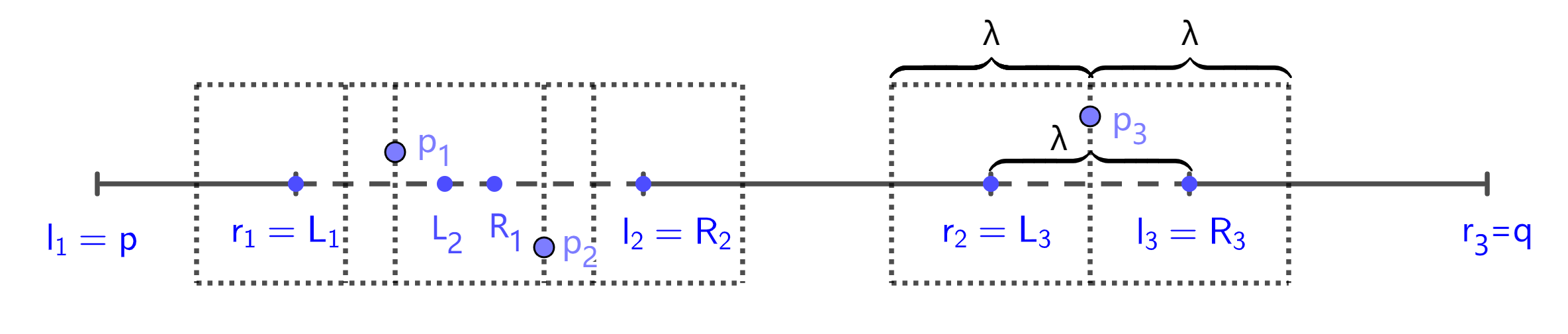}
  \caption{Obnoxious facilities at $p_1$, $p_2$ and $p_3$.$(L_1,R_1),(L_2,R_2),(L_3,R_3)$ are forbidden intervals that correspond to $p_1$, $p_2$ and $p_3$. $[l_1,r_1],[l_2,r_2],[l_3,r_3]$ are feasible intervals.}
  \label{fig:sq-1}
\end{figure}
\end{center}

For the $i$-th interval in $\Finterval$, denoted by $[l_i, r_i]$, the number of squares of size $\lambda$ that can be placed in it is 
$$
\text{Count}(i)=\lfloor (r_i - l_i) / \lambda  \rfloor + 1.
$$
Suppose $\Finterval$ contains $m$ intervals.
We can assume that the total number of squares that can be placed on $\overline{pq}$ is 
$\sum\limits_{i=1}^m \text{Count}(i)$ because the different feasible intervals do not interfere with each other (because each forbidden interval is of length $\lambda$). 
\footnote{
    Note that this is not true if we pack disks instead of squares. 
}
Thus, the time complexity of algorithm $A$ is $O(n)$ if we have the order of points in $P$; otherwise it is $O(n \log n)$.

\begin{theorem}
We can solve the \DSCOFL\ problem in exactly $O(n \log n)$ time.
\label{thm-21}
\end{theorem}

\subsubsection{Maximizing the Size of the Squares} \label{matrix}

In this subsection, we leverage the algorithm for the decision problem \DSCOFL\
to efficiently solve the optimization problem \DSCOFL.
We start with an easy observation regarding the optimal size $\lambda^*$.

\begin{observation}
\label{obs-1}
The optimal size $\lambda^*$ is divisible by the length of some feasible interval in $\Finterval$.
\end{observation}

Next, we sort all the (unsigned) $y$-coordinates of the points in $P$ as $Y = \{y_{p_1}, y_{p_2}, \cdots, y_{p_n}\}$ in increasing order. 
Using the decision algorithm, we can use binary search to determine the range of $\lambda^*$. 
In particular, we can find interval $[y_t, y_{t+1}]$ such that $\lambda^* \in [y_t, y_{t+1}]$.

Now, we know that $\lambda^* \in [y_t, y_{t+1}]$. Therefore, we can remove the point $p_i$ from $P$ such that $y_i > y_{t+1}$. We only have to consider the $x$-coordinate of the remaining points in $P$.
We have at most $n + 1$  feasible intervals in $\Finterval$, where we can place the center of the squares. Note that the number of feasible intervals could increase when $\lambda^*$ is reduced from $y_{t+1}$ to $y_t$ because certain overlapping forbidden intervals could become nonoverlapping when $\lambda^*$ decreases; other new feasible intervals could also appear during this processing. Therefore, we only consider the set of feasible intervals, denoted by $\Finterval_t$ when $\lambda^*=y_t$.

With Observation~\ref{obs-1}, we know that the length of at least one of the feasible intervals in $\Finterval_t$ can divide $\lambda^*$ in an optimal solution.
Therefore, we have at most $O(nk)$ candidates for $\lambda^*$ (at most $n+1$ feasible intervals
and placing at most $k$ centers in a feasible interval). 
In particular, for a feasible interval $(l_i, r_i)$, suppose we want to place exactly $t$ squares. We can calculate the maximum size of the square, denoted as $\lambda_{it}$, in $O(1)$ time. In particular, we can obtain $\lambda_{it}$ as the root of equation
$
\lambda_{it} = (r_i-l_i)/(t - 1)
$ (note that $l_i$ and $r_i$ are actually a function of $\lambda$).
These $\{\lambda_{it}\}$ are possible candidates for the optimal size $\lambda^*$.

We must apply the matrix search technique developed in \cite{2011Representing,chen2013algorithms}. 
This technique is useful for solving an optimization problem efficiently, 
if we have an efficient procedure to solve a certain decision problem 
(see \cite{chen2013algorithms,cao2015balanced} for different applications of this technique
to computational geometry problems).
We briefly recall their results.
Suppose there is a set of $M$ sorted arrays, such that the size of each array is at most $N$. We do not assume the set of arrays is stored explicitly. However, the value of each entry of an array can be evaluated in $O(1)$ time. $D(\lambda)$ is a decision procedure (called feasibility test) that takes a real value $\lambda$ as input and outputs either ``feasible" or ``infeasible". If $D(\lambda)$ returns ``feasible",
we also say $\lambda$ is a feasible value. An important property that $D$ satisfies is the following: if $\lambda$ is
a feasible value, any $\lambda'$ greater than $\lambda$ is also feasible.
Our goal here is to identify the smallest feasible value from these arrays
efficiently (in terms of the number of feasibility tests and execution time).
Formally, we have the following lemma proved in \cite{2011Representing}.

\begin{lemma}
\cite{2011Representing} Suppose we have a set of $M$ sorted arrays, such that each array's size is at most $N$, and each array element can be evaluated in O(1) time. Then, the smallest feasible value in these arrays can be computed by $O(\log (N + M))$ feasibility tests and the total execution time of the algorithm excluding
the feasibility test is $O(M \log N)$.
\label{lem-20}
\end{lemma}

We now demonstrate how to use the matrix search technique to solve our problem.
We have $n$ arrays of candidates for $\lambda^*$, each array of size $k$.
The value of the $t$-th entry of $i$-th array is $\lambda_{it}$.
Clearly, we can see $\lambda_{it} > \lambda_{i (t + 1)}$.
Hence, each array is sorted.
The feasibility test $D(\lambda)$ here is the decision problem to determine if the maximum number of squares of size $\lambda$ that can be placed is not greater than $k$.
The feasibility test can be implemented by the algorithm for \DSCOFL\ 
in $O(n\log n)$ time (Theorem~\ref{thm-21}).
It is easy to conclude that for any $\lambda>\lambda^*$, the answer to the decision problem is ``Yes",
and for any $\lambda<\lambda^*$, the answer is ``No".
Hence, the properties required by the matrix search procedure are satisfied and we can apply 
Lemma~\ref{lem-20}.

\begin{theorem}
The \SCOFL\ problem can be solved in $O(n \log n \log(n + k))$ time. 
\label{thm-20}
\end{theorem}

\subsection{\COFL}

Recall that in \COFL, our goal is to 
place $k$ facilities on $\overline{pq}$ to 
maximize the radius $\lambda$ defined as
$
\lambda=\min\{\max_{i\in [n]}\min_{j\in [k]} d(p_i, c_j), 
\alpha \cdot \max_{i\in [k-1]} d(c_i, c_{i+1}),
\}
$
where $\alpha>0$ is a given fixed coefficient.
In this section, we prove that
we can solve the \COFL\ problem in $O(n ^ 2 \log k)$ time. 
Similar to \SCOFL, we first solve the corresponding decision problem.
However, this decision problem is marginally more difficult than \SCOFL\ because, adjacent feasible intervals can interfere with each other.

\subsubsection{Decision Version: \DCOFL}

We first introduce the decision version \DCOFL\ defined as follows.
Given any value $\lambda>0$, \DCOFL\ asks if it is possible to place $k$ disks of radius $\lambda$ centered on $\overline{pq}$,
such that no points in $P$ lie inside any of these disks,
and the distance between any adjacent centers $d(c_i,c_{i+1})$ is at least $\lambda/\alpha$.
We use $\lambda^*$ to denote the optimal solution of the \COFL\ problem. 
We can conclude that \DCCOFL\ can determine if $\lambda > \lambda^*$.

Next, we present an algorithm $A$ to compute the maximum number of disks with radius $\lambda$ that can be placed. We denote this number by $A(\lambda)$,
and we can assume that if we can compute $A(\lambda)$ efficiently, we can solve the decision problem using the same time complexity.

Without loss of generality, we assume that segment $\overline{pq}$ is on the $x$-axis and $p$ is the origin $(0,0)$; we denote $q$ as $(x,0)$.

For each point $p_i$ in $P$, we can calculate the corresponding forbidden interval $f_i = [L_i(\lambda), R_i(\lambda)]$ such that a disk with radius $\lambda$ can cover point $p_i$ if and only if it is centered on $[L_i(\lambda), R_i(\lambda)]$. We denote the set containing all forbidden interval as $\mathcal{F}=\{f_i | i \in [n]\}$.

\begin{observation}
For a forbidden interval $s_i = [L_i(\lambda), R_i(\lambda)]$, $L_i(\lambda)$ is nonincreasing and $R_i(\lambda)$ is nondecreasing.
\label{obs-3.5}
\end{observation}

Excluding the forbidden intervals from $\overline{pq}$, we obtain the set of 
{\em feasible intervals} $\Finterval$. Notice that $\Finterval$ is a set of several disjoint intervals. Now, the $\DCOFL$ problem is equivalent to the problem of asking if one can place $k$ centers on feasible intervals in $\Finterval$ where the distance between any two centers is no less than $\lambda/\alpha$.

Next, we present the algorithm for the decision version, which consists of two steps.
In the first step, we indicate how to determine the set of feasible intervals $\Finterval$
efficiently. 
In the second step, we compute the maximum number of centers that we can place on $\Finterval$.

\paragraph{First step: Compute feasible intervals.}
Given any $\lambda$, we can obtain the value of $2n$ endpoints of $n$ forbidden
intervals in $O(n)$ time. After sorting these forbidden interval endpoints in $O(n\log n)$ time, we can scan the $2n$ ordered forbidden
interval endpoints and easily determine $t$ ($t \le n$) disjoint feasible intervals (sorted),
in  $\Finterval$. We denote the $i$-th feasible interval in $\Finterval$ as $[l_i(\lambda), r_i(\lambda)]$.

\paragraph{Second Step: Place the centers greedily.}
For simplicity of notation, we assume $\alpha=1$.
In general, $\alpha>0$ can be addressed in a similar manner.
Our goal is to determine if $k$ centers can be placed on $\Finterval$ 
such that the distance between every pair of nodes is no less than $\lambda$. We can solve this problem using Algorithm \ref{alg-3}.

\begin{algorithm}
	\caption{count$(\Finterval, n, \lambda)$} 
	\label{alg-3}
	\begin{algorithmic}[1]
	\State $\textsf{Res} \gets \lfloor (r_1(\lambda) - l_1(\lambda)) / \lambda \rfloor + 1$
	\State $\textsf{Last} \gets l_1(\lambda)+ (\textsf{Res} - 1) \cdot \lambda$
		\For {$i=2$ to $n$}
		\State $\textsf{Temp} \gets \lfloor (r_i(\lambda) - max(l_i(\lambda), \textsf{Last} + \lambda)) / \lambda \rfloor$
		\State $\textsf{Res} \gets \textsf{Res} + \textsf{Temp} + 1$
		\If{$\textsf{Temp} \ge 0$}
		    \State $\textsf{Last} \gets max(l_i(\lambda), \textsf{Last} + \lambda) + \textsf{Temp} \cdot \lambda$
		\EndIf
		\EndFor
	\State \textbf{return} \textsf{Res}	
	\end{algorithmic} 
\end{algorithm}

\begin{lemma}
Given that any $\lambda$, $A(\lambda)$ can be determined in $O(t \log k)$ time after we have computed $\Finterval$ in the first step, if the floor function is allowed, then we can solve the decision problem in $O(t)$.
\label{lem-2}
\end{lemma}
\begin{proof}
We place the center on the left endpoint of the first feasible interval. Then, we place centers greedily, which means that we place the next center that is closest to the previous center under the condition that we place it in the place that is at least $\lambda$ distance from the previous center on a feasible interval. 
The situation here is marginally more complicated than the square case because different feasible intervals could interfere with each other. The pseudocode can be found in Algorithm~\ref{alg-3}.

For the $i$-th ($i \ge 2$) feasible interval $[l_i(\lambda), r_i(\lambda)]$ in $\Finterval$, suppose the $x$-coordinate of the previous feasible interval is 
$\textsf{Last}_{i-1}$. We define the floor function of a negative real number $x$ as the greatest integer less than $x$ ( e.g., $\lfloor -3.2\rfloor = -4$). Then, the number of centers that can be placed on the $i$-th feasible interval $(l_i,r_i)$, 
denoted by $\text{Count}(i)$, is computed as
$$
\text{Count}(i)=\lfloor (r_i - \max(l_i, \textsf{Last}_{i-1}+\lambda)) / \lambda \rfloor + 1.
$$
For the first feasible interval in $\Finterval$, 
$\text{Count}(1)=\lfloor (r_1 - l_1) / \lambda \rfloor + 1$.
Finally, we obtain $A(\lambda) = \sum\limits_{i=1}^t \text{Count}(i)$.

$A(\lambda)$ can be determined in $ O(t \log t)$ time if the floor function is allowed (which requires $O(1)$ time) after the first step, and otherwise in $ O(t \log t + t \log k)$ time ($0\le \text{Count}(i) \le k$).
\end{proof}

We summarize our result in the following theorem.
\begin{theorem}
We can solve the decision problem 
\DCOFL\ in $O(n \log n)$ time if the floor function is allowed (which requires $O(1)$ time); otherwise $O(n \log n + n \log k)$ time.
\end{theorem}

\subsubsection{Maximizing the Radius}

In this subsection, we use the algorithm for the decision version \DCOFL\
to design an efficient algorithm for the optimization problem \COFL.
Without loss of generality, we assume $\alpha=1$.

\paragraph{Method}
We have $t$ sorted disjoint feasible intervals where we can place the centers.
Each endpoint of these $t$ disjoint intervals is a polynomial of the radius $\lambda$ of the disk. We denote the $i$-th feasible interval as $(l_i(\lambda), r_i(\lambda))$. We want to determine the maximum $\lambda^*$ such that we can place $k$ nodes on those intervals under the condition that the distance between every pair of nodes is not less than $\lambda^*$. 

Notice that if we place a node on the $i$-th interval, it could influence the number of nodes that we can place on the $(i+1)$-th interval when the distance between the right endpoints of the $i$-th interval andleft endpoint of the $(i+1)$-th interval is less than $\lambda^*$. Similar to Observation~\ref{obs-1}, we have the same observation regarding
$\lambda^*$.

\begin{observation}
For the optimal radius $\lambda^*$, we can find at least one pair of endpoints $l_i(\lambda)$ and $r_j(\lambda)$ in $\Finterval$,
such that the distance between $r_j(\lambda)$ and $l_i(\lambda)$ can be divided by $\lambda^*$.
\label{obs-2}
\end{observation}

Again, we apply the matrix search technique developed in \cite{2011Representing},
which we reviewed in Section~\ref{matrix}.

The reduction of our optimization problem to a matrix search problem is similar to that in Section~\ref{matrix}.
For each pair $l_i(\lambda),r_j(\lambda)\in S$, we can generate a sorted list as $\{\lambda_{1}, \lambda_{2}, \cdots, \lambda_{k}\}$, where $\lambda_t$ denotes the root of the equation $r_j(\lambda) - l_i(\lambda) = t \lambda $. 
We note that $r_j(\lambda) - l_i(\lambda)$ is non-increasing in $t$, based on Observation~\ref{obs-3.5}.
Therefore, $\lambda_t \ge \lambda_{t+1}$. 
Now, we have a set of $t^2$ sorted arrays where each array's size is $k$.
The optimal $\lambda^*$ is the value of an entry of an array.
We use our algorithm for \DCOFL\ as the feasibility test.
Finally, $\lambda^*$ can be found by applying Lemma \ref{lem-20}. We summarize our result in the following lemma.

\begin{lemma}
Given $t$ sorted disjoint intervals where each endpoint of these $t$ disjoint intervals is a polynomial of the diameter $\lambda$ of the disk. The problem of finding the maximum $\lambda$ such that we can place $k$ nodes under the condition that the distance between each pair of nodes is no less than $\lambda$ can be solved in $O(t^2 \log k + t \log k \log  (t^2 + k))$ time.
\end{lemma}

Next, we can conclude our final theorem for the \COFL\ problem.
\begin{theorem}
The \COFL\ problem can be solved in $O(n^2 \log k + n \log k \log  (n^2 + k))$ time.
\end{theorem}

\subsection{\CCOFL}

Recall that in \CCOFL, we are given a set $P=\{p_1, p_2, \cdots, p_n\}$ of $n$ demand points in the plane, a predetermined circle $\circle$ with radius $r_c$ and a positive integer $k$.
We must locate $k$ facility sites $C=\{c_1,\ldots,c_k\}$ on the boundary arc $\partial \circle$.
Our goal is to maximize the Euclidean distance from any demand point in $P$ to its closest facility and the mutual distance between any two adjacent facilities. 
Formally, our goal is to maximize
$
\min\{\max_{i\in [n]}\min_{j\in [k]} d(p_i, c_j), 
\alpha\cdot\max_{i\in [k]} d(c_i, c_{i+1})
\}
$
, where $c_{k+1}$ is understood as $c_1$.
We first solve the corresponding decision problem efficiently.


\begin{center}
 \begin{figure}[t]
  \centering    
  \includegraphics[scale=0.16]{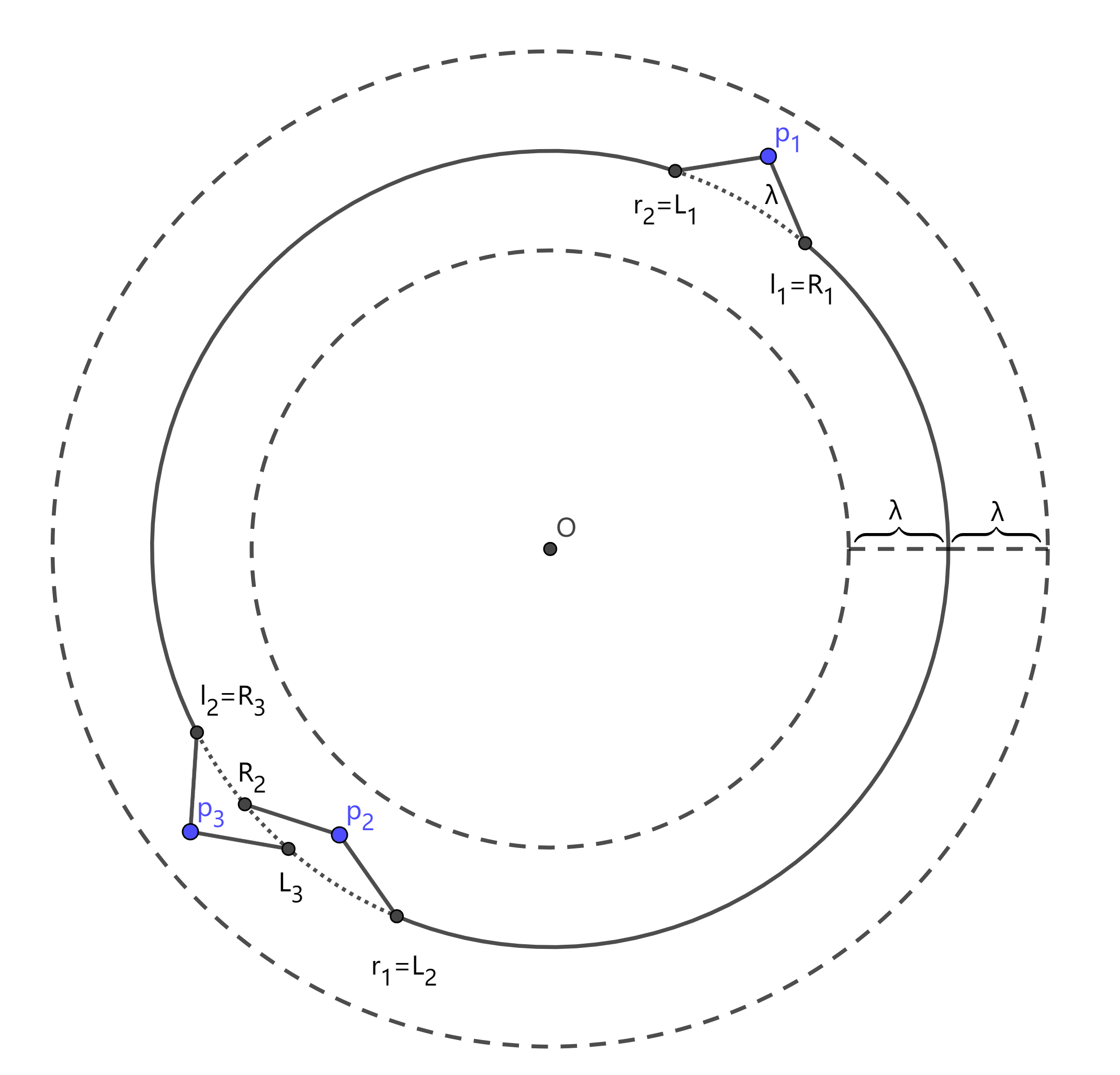}
  \caption{$(L_1,R_1),(L_2,R_2),(L_3,R_3)$ are forbidden intervals corresponding to obnoxious facilities $p_1,p_2,p_3$. $[l_1,r_1], [l_2,r_2]$ are feasible intervals.}
  \label{fig:birds}
\end{figure}
\end{center}

\subsubsection{Decision Version: \DCCOFL}
We first consider the decision version of \DCCOFL.
In the decision version, we are given a certain value $\lambda>0$ and we must determine if the maximum number of facilities that we can place is no less than $k$
such that the distance between each pair of facilities and any facility to any node is not less than $\lambda$.
The \DCCOFL\ problem is similar to the \DCOFL\ problem. However, we cannot directly use the greedy algorithm because we do not know where we should place the first center. 
In approximate terms, 
if a problem can be solved greedily on a line segment in $O(T)$ time,
it can be solved on a circle in $O(nT)$ time, by testing $n$ different starting points
in the circle and reducing it to the problem on a line.
Hence, if we enumerate where to place the first center and execute the greedy algorithm (Algorithm~\ref{alg-3}),
the execution time is $O(n^2\log k)$.
In this section, we present a more efficient algorithm.
First, we formulate the problem as follows.


\begin{definition} \DCCOFL.

Given $n$ disjoint feasible intervals $\Finterval=\{[l_1,r_1], \cdots, [l_n,r_n]\}$ on a circular ring, we must determine the maximum number of centers that we can place on $\Finterval$, denoted by $k^*$, such that the Euclidean distance between each pair of centers is no less than $\lambda$. Return true if $k^* \ge k$; otherwise, return false.
 \end{definition}
 
To simplify the analysis, we assume that the float division can be accomplished in $O(1)$ time.
We state that a placement of centers is {\em optimal} if no other placement can place more centers.
 
\begin{observation}
We can find at least one optimal placement such that we can place a node on one endpoint of an interval in $\Finterval$.
\end{observation}
 
 
\begin{observation}
If the length of a forbidden interval is at least $\lambda$, we can directly reduce the \DCCOFL\ problem to the \DCOFL\ problem.
\end{observation}
 
\begin{proof}
If the length of a forbidden interval $[l,r]$ is not less than $\lambda$, clearly at least one optimal solution can place a node on $l$ or $r$. Hence, we can cut the circular ring here and the problem is the same as the \DCCOFL\ problem.
\end{proof}
 
Therefore, we can assume that in the \DCCOFL\ problem, the length of forbidden intervals
is less than $\lambda$.
 
Without loss of generality, consider placing the first center on an endpoint (e.g., $r_i$). We place centers one by one $\lambda$ distance away (e.g., in the clockwise direction) until the center that we place lies in one of the forbidden intervals. Suppose the corresponding forbidden interval is $[l_j,r_j]$.
Then, we place the center at $r_j$. We call this process a {\em jump} from $i$ to $j$ and denote this jump as $N[i][0] = j$. We also denote $C[i][0]$ as the number of centers that we place during this process. We also use $N[i][j]$ to denote that after $2^j$ jumps, $i$ jumps to $N[i][j]$ (See Figure~\ref{fig:jump} for an example).
 
 \begin{figure}[h!]
  \includegraphics[scale=0.3]{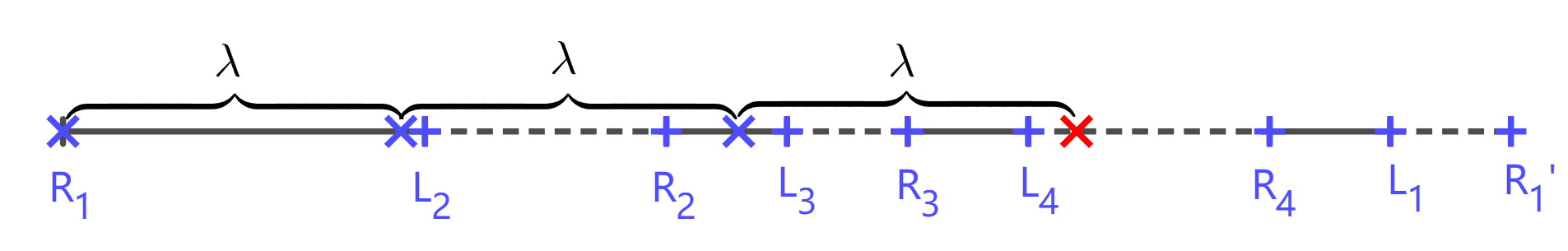}
  \caption{Illustration of a jump. To start, we place a center at $R_1$.
  If we place the centers one by one $\lambda$ distance away, the 4th
  center would lie in the 4th forbidden interval (the red cross). 
  Therefore, we must place it at $R_4$. Hence, $N[1][0]=4$.}
  \label{fig:jump}
\end{figure}
 
 \begin{paragraph}{Persistent segment trees.}
 We require a segment tree data structure, which we briefly review here (see, e.g., \cite{bentley1980optimal}). 
 A segment tree is a data structure that stores the information in an array as a tree. This structure allows efficient answering of range queries over an array, and allows quick modification of the array. It supports finding the minimum or sum of any range of consecutive array elements in  
 $O(\log n)$ time, where $n$ is the length of the array. It also allows us to modify the array online by adding a value to an array element or modifying the values of a range 
 (e.g., assigning a value to all elements, or adding a value to all elements in the range). 
 We can build a segment tree in $O(n)$ time.  
 
 We require a persistent version of the segment tree, to efficiently query intermediate versions of 
 the segment tree.
 A persistent data structure is a data structure that preserves its previous version when it is modified. A data structure is partially persistent if all versions can be accessed yet only the newest version can be modified. For a detailed introduction of persistent data structure, see \cite{kaplan2018persistent}.
 
 We can apply partial persistency in the segment tree and also ensure that it requires $O(\log n)$ time for each query or modification.
 
 For an array $A[i](1\le i\le n)$, we define the \textbf{ADD} and \textbf{QUERY} operations as follows:
 \begin{enumerate}
     \item \textbf{ADD}(i,j): increase all values of $A[x]$ ($i\le x\le j$) by one.
 
     \item \textbf{QUERY}(i,t): return the value of $A[i]$ after $t$ ADD operations.
  \end{enumerate}
 
Using the persistent segment tree data structure, we can support \textbf{ADD} and \textbf{QUERY} operations in $O(\log n)$ time. 
Next, we confirm that we can compute $N[i][0]$ efficiently using the persistent segment tree.
 
 \end{paragraph}
 
\begin{theorem}
We can compute the values of $N[i][0]$ for $1\le i \le n$ in $O(n \log ^ 2 n)$ time.
\end{theorem}

\begin{proof}
Suppose $[l_1,r_1]$ is a feasible interval and we set $l_1$ as the starting point. 
We use $\%$ to denote a float mod operation, which means for any nonnegative real number $a,b$, $a \% b = a - \lfloor a / b \rfloor \times b$. We transform an interval $[l_i,r_i]$ in $\Finterval$ to 
interval $[l_i',r_i']$, where $l_i' = |l_i - l_1| \% \lambda$ and $r_i' = |r_i - l_1| \% \lambda$. 
Note that if $l_i' > r_i'$, this interval actually means two intervals, $[r_i',\lambda]$ and $[0,l_i']$.

The key observation is that there is a jump from $i$ to $j$ (i.e., 
$N[i][0]=j$) iff 
$l_i \% \lambda \in [l_j',r_j']$ 
(or $l_i+k\lambda \in [l_j,r_j]$ for a positive integer $k$)
and $j$ is closest to $i$ among these indices ($j$ is the first such point).
Hence, to compute $N[i][0]$, we must determine the minimum $j$ such that $l'_i \in (l'_j, r'_j)$.

Next, we demonstrate how to compute $N[i][0]$s using a persistent segment trees.
We first sort all $l_i'$ and $r_i'$ for $1\le i \le n$ and store it in an array $A$. We define a function $\Index(x)$ where $x$ is either a $l_i$ or $r_i$ and 
$\Index(x)$ returns the index of $x$ in $A$.
Note that $A[\Index(x)] = x$.

We build a persistent segment tree to maintain array $A$. 
To insert an interval $[l_i',r_i']$, we perform \textbf{ADD}($\Index(l_i'),\Index(r_i')$).
We insert the intervals $[l_1',r_1'],\cdots, [l_n',r_n']$ into the persistent segment tree
one by one in this order, and we repeat this operation (because of the circularity). 

For any $i$, to determine $N[i][0]$, we perform a binary search on $t$ to find the maximum $t^*$ such that the result of \textbf{QUERY}($\Index(r_i'),t^*)$) equals that of \textbf{QUERY}($\Index(r_i'),i$). Then, $N[i][0] = t^*$.

Finally, we calculate the execution time. We insert $2n$ intervals in $O(n\log n)$ time.
Computing each $N[i][0]$ for any $1\le i \le n$ costs $O(\log^2 n)$ time (each binary search step
is a \textbf{QUERY} that requires $O(\log n)$ time). Hence, the overall time complexity is $O(n \log^2 n)$.
\end{proof}

Now, we have $N[i][0]$ for $1\le i \le n$. Then, we can do the binary lifting on $N$ using Algorithm \ref{alg-1} to compute $N[i][j]$ for a greater $j$. 

Note that $C[i][j]$ is the number of centers that we place during $i$ using $2^j$ jumps. Therefore, $C[i][0]$ is the number of centers that we can place between $r_i$ and $N[i][0]$. We denote the distance between $r_i$ and the left endpoints corresponding to $N[i][0]$ as $ d(\textsf{index}(r'_i),N[i][0] - 1)$ and $C[i][0] = \lfloor d(\textsf{index}(r'_i),N[i][0] - 1) / \lambda \rfloor $.

\begin{algorithm}
	\caption{Binary\_Lifting$(N, n, m)$} 
	\label{alg-1}
	\begin{algorithmic}[1]
		\For {$i=1$ to $2n$}
		\State $ \textsf{C[i][0]} \gets \lfloor d(\textsf{index}(r_i),N[i][0] - 1)  / \lambda \rfloor + 1 $
		\EndFor
		\For {$i=1$ to $2n$}
			\For {$j=1$ to $m$}
				\State $\textsf{N[i][j]} \gets \textsf{N[N[i][j - 1]][j-1]}$
				\State $\textsf{C[i][j]} \gets \textsf{C[i][j-1]} + \textsf{C[N[i][j - 1]][j-1]}$
			\EndFor
		\EndFor
	\end{algorithmic} 
\end{algorithm}

Now, we can obtain the number of nodes that we can place if we begin by placing a node on $l_i$, denoted as \textsf{num}, in $O(\log n)$ time, as indicated in Algorithm \ref{alg-2}. 

We must simulate the placing process by jumping. Note that we require at most $k$ jumps because we place at least one center during a jump. Now, let $step = \lfloor \log k
\rfloor$. If after $2^\textsf{step}$ jumps, $x$ jumps over $x + n$, meaning that the number of jumps that we require is less than $2^\textsf{step}$. Therefore, we reduce the \textsf{step} by one and make the next attempt. Otherwise, we simply take this $2^\textsf{step}$ jumps and make the next attempt. This process terminates when $\textsf{step} < 0$.


\begin{algorithm}
	\caption{Cal$(x, n, k)$} 
	\label{alg-2}
	\begin{algorithmic}[1]
	    \State $\textsf{step} = \lfloor \log k
\rfloor$
	    \State $\textsf{target} = x + n$
	    \State $\textsf{num} = 0$
		\While {step $\ge 0$}
			\If {$\textsf{N[x][step]} \le \textsf{target}$}
				\State $x \gets \textsf{N[x][step]}$
				\State $\textsf{num} \gets \textsf{num} + \textsf{C[x][step]}$
			\EndIf
			\State $\textsf{step} \gets \textsf{step} -1$
		\EndWhile
	\State $\textsf{num} \gets \textsf{num} + \lfloor \textsf{d(x, target)}  / \lambda \rfloor$ 
	\State \textbf{return} \textsf{num}
	\end{algorithmic} 
\end{algorithm}

\begin{theorem}
We can solve the decision version \DCCOFL\ in $O(n \log^2 n + n \log k)$ time.
\end{theorem}
 
\subsubsection{Maximizing the radius}

Next, we solve the  \CCOFL\ optimization problem.
Again, we can solve the \CCOFL\ problem using the same matrix search technique used for the \COFL\ problem. 
For each pair of left endpoints of a forbidden interval $l_i(\lambda)$ and $r_j(\lambda)$ as a right endpoint of a forbidden interval, we can generate a sorted list as $\{\lambda_{1}, \lambda_{2}, \cdots, \lambda_{k}\}$, where $\lambda_t$ denotes the root of the equation $r_j(\lambda) - l_i(\lambda) = t \lambda $ and $\lambda_t \ge \lambda_{t+1}$ holds.  Then, we use Lemma \ref{lem-20} to find $\lambda^*$ using $\log(n^2 + k)$ feasibility tests 
(Algorithm~\ref{alg-2} of the \CCOFL\ problem).

Then, we can state the result for the \CCOFL\ problem.
\begin{theorem}
The \CCOFL\ problem can be solved in $O(n ^ 2 \log k + n (\log^2 n + \log k) \log(n^2 + k) )$ time.
\end{theorem}

\section{Minsum Obnoxious Facility Location Problem}

In this section, we discuss the \MOFL\ problem.
We solve the \MOFL\ problem by reducing it to the $k$-link shortest path problem.

\paragraph{$k$-link shortest path:}
Let $G = (V, E)$ be an edge weighted, complete, DAG with the
vertex set $V = \{v_l, v_2, \cdots ,v_n\}$.
For $1 \le i <j \le n$, we use $w(i,j)$ to denote the weight of the edge $(i,j)$.
For a path that contains exactly $k$ links (i.e., edges), we call the path a $k$-link path.
For any two vertices $i,j$, the minimum $k$-link path from $i$ to $j$ is the path from $i$ to $j$  that contains exactly $k$ links and has the minimum total weight among all such paths.

If the weights satisfy the following convex Monge property, the minimum $k$-link short path problem
can be solved more efficiently \cite{aggarwal1993finding}.
For a weighted DAG $G$, $G$ satisfies the convex Monge property if for all $1 \le i <j \le n$, the inequality $w(i,j) + w(i + 1, j + 1) \ge w(i, j + 1) + w(i + 1, j)$ holds.

\begin{theorem} \cite{aggarwal1993finding}
The minimum k-link path that satisfies the convex Monge property can be solved in $O(n \alpha(n) \log ^ 3 n)$ time. 
\label{thm-3}
\end{theorem}

Next, we demonstrate that \MOFL\ can be reduced to the $k$-link shortest path problem with 
convex Monge property. To do this, we first demonstrate the reduction for a simplified version of the \MOFL\ problem, denoted by \SMOFL, 
where we remove the requirement that the distance between two centers is at least $\alpha \lambda$.

\subsection{\SMOFL\ problem}
For each point $p_i$ in $P$, we can calculate the corresponding influence interval $f_i = (l_i, r_i)$ such that a disk with radius $\lambda$ can cover point $p_i$ if and only if it is centered on $(l_i, r_i)$. We denote the set containing all influence intervals as $F=\{f_i |  i \in [n]\}$. For ease of notation, without loss of generality, we suppose all $l_i$ and $r_i$ ($i \in [n]$) are distinct, $w_i > 0$($i \in [n]$) and $p < l_i \le r_i < q$($i \in [n]$).

\begin{definition}\SMOFL\ problem. 

We are given a set of $n$ influence intervals $F=\{f_i |  i \in [n]\}$ (in increasing order), an integer $k$, and a segment $\overline{pq}$. Each interval in $F$ is assigned with a weight $w_i$. Our goal is to place $k$ centers $p_1,\cdots,p_k$ on $\overline{pq}$ such that 
$$
\sum_{j=1}^k \sum_{\{i : \exists i,p_i \in (l_j,r_j)\}} w_j
$$
is minimized, i.e., the total weight of the intervals on which the centers lie is minimized.

\end{definition}

\begin{observation}
\label{obs-123}
There exists an optimal solution for the \SMOFL\ problem such that all the centers are placed on the endpoints in $F$.
\end{observation}

Next, we reduce this to the $k$-link shortest path problem. We build DAG $G = (V, E)$ as follows. $V$ is the set that contains all the endpoints in $F$ and endpoints of segment $\overline{pq}$, i.e., $V=\{l_1,r_1,l_2,r_2,\cdots,l_n,r_n,p,q \}$. We index the nodes in $V$ based on the coordinates of the node in ascending order and denote the index of node $x \in V$ by \textsf{index(x)}($1\le \textsf{index(x)}\le 2n + 2$). $E$ contains all the edges $x \rightarrow y$ such that $1\le x < y \le 2n + 2$. The weight for edge $x \rightarrow y$ is the opposite value of the total weight of the influence intervals between node $x$ and node $y$, i.e., $w(x, y) = -\sum\limits_{j=1}^k w_j \cdot [x \le \textsf{index($l_j$)} \le \textsf{index($r_j$)} \le y] $. In particular, $w(1, x) = w(x, 2n + 2) = 0(2 \le x \le 2n + 1)$. The nodes in the k-link shortest path are where we place the obnoxious facilities(except $p$ and $q$). We define the total weight of the influence intervals in $F$ as the \textsf{sum}. The minimum total weight of the $(k+1)$-link path of $G$ from $p$ to $q$ adding \textsf{sum} is the answer of the corresponding \SMOFL\ problem. 

Next, we prove that $G$ satisfies the convex Monge property.

\begin{theorem}
\label{thm-123}
$G$ satisfies the convex Monge property.
\end{theorem}
\begin{proof}
To demonstrate that graph $G$ satisfies the convex Monge property, we must prove $$w(i,j) + w(i + 1, j + 1) \ge w(i, j + 1) + w(i + 1, j)$$

We define $W(i,j)$ as the total weight of the influence intervals $f_i \in F$ such that $\textsf{index($l_i$)} = i$ and $\textsf{index($r_i$)} \le j$. Next, we can obtain 
$$w(i,j) = w(i+1,j) + W(i,j)$$
$$w(i,j + 1) = w(i+1,j+1) + W(i,j+1)$$

clearly,$W(i,j) \ge W(i, j+1)$ (note that the weight of each influence interval in $F$ is negative).
Therefore, we obtain $w(i,j) + w(i + 1, j + 1) \ge w(i, j + 1) + w(i + 1, j)$.
\end{proof}

\subsection{Solving the \MOFL\ problem}
The only difference between the \SMOFL\ problem and the \MOFL\ problem 
is that the distance between each pair of centers should be at least $\alpha \lambda$ in the \MOFL\ problem. 

It is possible that we cannot determine an optimal solution by only placing
centers on endpoints in $F$ under the condition that the distance between each pair of centers is least $\alpha \lambda$ (see Figure \ref{fig:place} for an example). 
Hence, we require the following lemma.

\begin{lemma}
We can find an optimal solution such that we only place centers on endpoints in $F$ and the place such that the distance from this place to an endpoint in $F$ can be divided by $\alpha \lambda$, e.g., $l_i + k\alpha \lambda$. 
\end{lemma}
\begin{proof}
Note that for any solution to the \MOFL\ problem, the answer only changes when we move a center out of or into an influence interval. For an optimal solution, suppose $c_i$ is the first center that does not place in the place mentioned above. Then, we can move $c_i$ to the leftmost position such that it is either on an endpoint of influence interval in $F$ or simply $k \alpha \lambda$ (for an integer $k>0$) distance after the previous center without changing the answer.
\end{proof}

Therefore, for each position mentioned above, we add $k-1$ extra position that just $\alpha \lambda$ after it one by one (because we place $k$ centers totally). For example, for $l_i$ in $F$, we also add $l_i+\alpha \lambda$, $l_i+2\alpha \lambda,\cdots,l_i+(k-1)\alpha \lambda$ to $V$. 

To maintain the distance between each pair of centers at least 
$\alpha \lambda$, if the distance between node $x$ and node $y$ is greater than $\alpha \lambda$, we set  $w(x, y) = -\infty$.

Finally, we reduce the \MOFL\ problem to the k-link shortest path problem with $O(nk)$ nodes. We demonstrate the result in the following theorem.

\begin{figure}[h!]
  \includegraphics[width=1\textwidth]{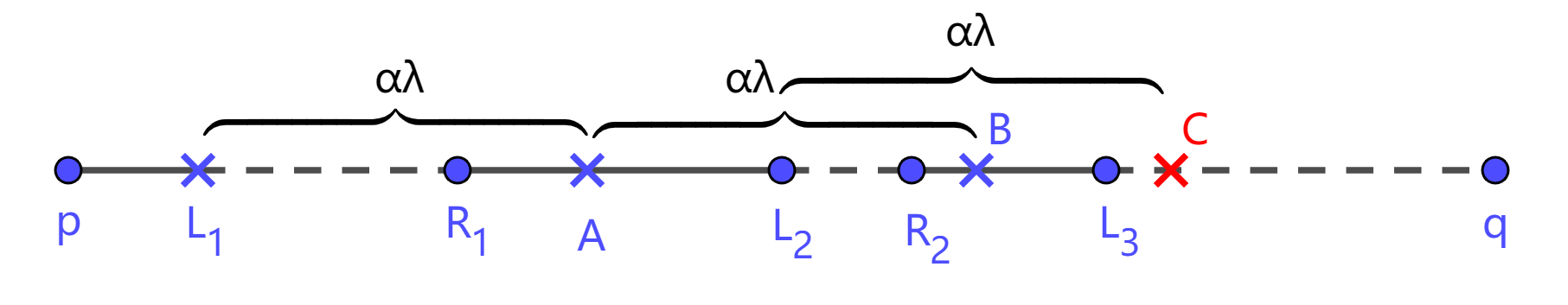}
  \caption{If we only place centers on endpoints in $F$, we can only place two centers ($L_1$ and $L_2$). The optimal solution can place three centers (i.e., $L_1$, $A$, $B$).}
  \label{fig:place}
\end{figure}

\begin{theorem}
We can solve the $\MOFL$ problem in $O(nk\cdot \alpha(nk) \log^3 {nk})$ time using the k-link shortest path.
\label{thm-4}
\end{theorem}

\section{Conclusion}
In this thesis, we studied four versions of the obnoxious facility problem restricted to a line segment or circle and obtained improved results for these problems. We provided an efficient solution for the decision versions of \COFL, \SCOFL, and \CCOFL. Using this, we improved the time complexity for solving \COFL, \SCOFL, and \CCOFL. We also improved the time complexity of the \MOFL\ problem by using the k-link shortest path. 

Our results for \COFL, \SCOFL, \CCOFL, and \MOFL\ obtained the best known results for these problems.
We expect to see more applications of and further work on our results.

\section{Acknowledgements}
The author would like to thank Jian Li and Haitao Wang for several helpful discussions.
The research is supported in part by the National Natural Science Foundation of China Grant 62161146004, Turing AI Institute of Nanjing and Xi'an Institute for Interdisciplinary Information Core Technology.

\bibliographystyle{IEEEtran}
\bibliography{ref}
\end{document}